\documentclass[preprint,12pt]{elsarticle}

\usepackage{amssymb}
\usepackage{graphicx,subfigure}
\usepackage{epsfig}
\usepackage{amsmath,environ}
\usepackage{mathrsfs}
\usepackage{color}
\usepackage{amsthm}
\usepackage{url}
\usepackage{lineno}

\newtheorem{theorem}{Theorem}[section]
\newtheorem{definition}{Definition}[section]

\newtheorem{proposition}[theorem]{Proposition}
\def\Pcum{{\cal{P}}_{\mbox{\scriptsize \rm{cum}}}}

\journal{JSTAT}

\begin{document}

\begin{frontmatter}

%% Title, authors and addresses

%% use the tnoteref command within \title for footnotes;
%% use the tnotetext command for theassociated footnote;
%% use the fnref command within \author or \address for footnotes;
%% use the fntext command for theassociated footnote;
%% use the corref command within \author for corresponding author footnotes;
%% use the cortext command for theassociated footnote;
%% use the ead command for the email address,
%% and the form \ead[url] for the home page:
%% \title{Title\tnoteref{label1}}
%% \tnotetext[label1]{}
%% \author{Name\corref{cor1}\fnref{label2}}
%% \ead{email address}
%% \ead[url]{home page}
%% \fntext[label2]{}
%% \cortext[cor1]{}
%% \address{Address\fnref{label3}}
%% \fntext[label3]{}

\title{Corona graphs as a model of small-world networks}

%% use optional labels to link authors explicitly to addresses:
%% \author[label1,label2]{}
%% \address[label1]{}
%% \address[label2]{}

\author{Qian Lv, Yuhao Yi and Zhongzhi Zhang\corref{corZhang}}

\address{ \fnref{lab}Shanghai Key Laboratory of Intelligent Information
Processing, School of Computer Science, Fudan University, Shanghai 200433, China}
\cortext[corZhang]{Corresponding author. E-mail address: zhangzz@fudan.edu.cn}
\begin{abstract}
%% Text of abstract
We introduce recursive corona graphs as a model of small-world networks. We investigate analytically the critical characteristics of the model, including order and size,  degree distribution, average path length, clustering coefficient, and the number of spanning trees, as well as  Kirchhoff index. Furthermore, we study the spectra for the adjacency matrix and the Laplacian matrix for the model. We obtain explicit results for all the quantities of the recursive corona graphs, which are similar to those observed in real-life networks.
\end{abstract}

\begin{keyword}
Corona product \sep Small-world \sep Spectra \sep Laplacian matrix
%% keywords here, in the form: keyword \sep keyword

%% PACS codes here, in the form: \PACS code \sep code

%% MSC codes here, in the form: \MSC code \sep code
%% or \MSC[2008] code \sep code (2000 is the default)

\end{keyword}

\end{frontmatter}

%% \linenumbers

%% main text

\section{Introduction}

For decades we want to know what a graph looks like. We want to reveal the principles of the networks' behaviour covered by their complex topology and dynamics. We want to learn about how the network structure evolves over time and how it affects the properties of dynamical processes on it. For the nature of decentrality of real networks, it is hard to observe the networks directly. Instead we observe them by taking snapshots of their network structure and content and keep updating them. Yet this gives little information about the future, since many of them keep growing over time. Thus it is desirable to set up models to fit real networks in both structure and functionality. We can use the models to mimic real-life networks. We also expect that the properties of the models can be proven rigorously, thus we can find the relations between topological and dynamical properties of networks. Even if it is hard to give closed-formed expressions for some quantities, it would be nice to make them tractable for convenience of estimation.

Among various network models, the ER graph proposed by Erd\"{o}s and R\'{e}nyi~\cite{ErRe60} is the earliest one. It generates a random graph by choosing a constant probability for joining every pair of vertices in the network. The model exhibit interesting statistical properties and has been well studied by many people. However the model lacks some important properties of real-world networks. For example many real-world networks exhibit small-world property~\cite{Mi67,WaSt98,Kl00STOC} with their diameters growing logarithmically in the number of  vertices, while maintaining a high clustering coefficient. The Watts-Strogatz (WS) model~\cite{WaSt98} is a typical graph model with the small-world effect. %It is also intensely studied and many properties have been revealed and verified through out the recent two decades. 
Nevertheless, because of its randomness, many of its properties cannot be derived precisely, for example eigenvalues of the adjacency and Laplacian matrices. Thus deterministic models are often used to mimic complex networks~\cite{ZhCo11}, since their structural~\cite{ZhWu15} and spectral~\cite{LiDoQiZh15} characteristics  can be determined analytically. In addition to the small-world effect, another important feature of a network is degree distribution. Many real-world networks exhibit a heavy-tail distribution while some networks have an exponential distribution~\cite{Ne03,BaWe00,iJaSo01}. The famous preference attachment~\cite{BaAl99,AlBa02,KlKuRaRaTo99,KuRaRaSrTo99} scheme successfully described the growing process of networks with heavy-tailed degree distribution. This work, like the WS model~\cite{WaSt98}, leads to a network with an exponential degree distribution.

Recently graph (matrix) products have been applied to modelling graphs with the same properties as real-life networked systems, such as the Cartesian product~\cite{ImKl00}, dot product~\cite{YoSc07} and Kronecker product~\cite{We62,LeFa07,Ma07,Le09,LeChKlFaGh10}. A merit of such methods is that the graph/matrix products facilitate estimate of the properties of the generated graphs. 

In this paper we introduce a recursive way to generate small-world networks with an exponential degree distribution, based on corona product of graphs. We obtain exact solutions to many structural properties of the networks. Moreover, we derive all the eigenvalues for their adjacency matrix and Laplacian matrix, which are provided in a recursive way. Based on the obtained eigenvalues, we calculate the number of spanning trees, as well as the Kirchhoff index of the networks. %The network shows small Kirchhoff index that average effective distance between two vertices grow logarithmicly with the order of the graph. The paper is self-contained and the results are verified across different methods.

\section{Graph Construction}
%Let $G=(V(G),E(G))$ be an undirected graph with vertex set $V(E) = \{1, 2, \dots , n\}$ and edge set $E(G) = \{1, 2, \dots , m\}$.
% By convention the (unweighted) adjacency matrix $A(G)$ is defined as a $N\times N$ matrix with entry $a_{i,j}$ be the number of edges in $G$ with endpoints ${i,j}$. And the degree matrix $D(G)$ is defined as a diagonal matrix with its $i$th entry on the main diagonal equal to the degree of that vertex. And we call $L(G)=D(G)-A(G)$ the Laplacian matrix of graph $G$. These matrices determine the structure of the graph, and the eigenvalues of $A(G)$ and $L(G)$ are subtle to many of the structural properties, which have remarkable impact on the dynamic processes superimposed upon the network.

In this paper, we use corona product to generate a small-world network model. Literatures about the corona product and its related graphs are partly established~\cite{BaPaSa07,Li14}. Let $G=(V(G),E(G))$ be the embedded graph of a network. Suppose the graph is undirected and has vertex set $V(G) = \{1, 2, \dots , N\}$ and edge set $E(G)\subseteq V(G)\times V(G)$. We define the number of vertices $N$ as the order of the graph, and the number of edges as the size of the graph, denoted as $M=|E(G)|$.

Given two graphs $G_1$ and $G_2$ , their corona product $G_1\circ G_2$ is defined as follows.
\begin{definition}\label{coronaDef}
Let $G_1$ and $G_2$ be two graphs with disjoint vertex sets. $G_1$ has $N_1$ vertices and $G_2$ has $N_2$ vertices. Their corona product $G_1\circ G_2$ is a new graph which consists of one copy of $G_1$ nd $N_1$ copies of $G_2$. The $i$-th vertex of $G_1$ is joint by a new edge with every vertex in the $i$-th copy of $G_2$.
\end{definition}

In this paper we investigate the case where $G_2$ is the $q$-complete graph, thus we give the definition of the recursive corona graph.

\begin{definition}\label{recCorDef}
Let $K_q$ be the $q$-complete graph ($q\geqslant2$), then the $g$th generation of recursive corona graph (RCG) $C_q(g+1)$ is defined as the corona product of the previous generation of RCG $C_q(g)$ and $K_q$. More formally, $C_q(g+1)$ is defined as $C_q(g+1)=C_q(g)\circ K_q$, $g \geqslant 0$, with the initial condition $C_q(0)=K_q$.
\end{definition}

Figure~\ref{examples} illustrates the construction process for a particular network $C_4(g)$.

%%%%%%%%%%%%%%%%%%%%%%%%%%%%%%%%%%%%%%%%%%%%%%%%%%%%%%%%%
% Figure  1
%%%%%%%%%%%%%%%%%%%%%%%%%%%%%%%%%%%%%%%%%%%%%%%%%%%%%%%%%%
\begin{figure}[htbp]
\begin{center}
\includegraphics[width=.5\linewidth]{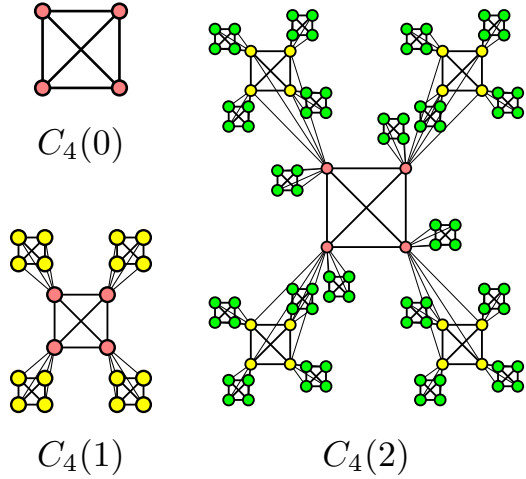}
\end{center}
\caption[kurzform]{Construction of the recursive corona graph, showing the first three generations of $C_4(g)$}\label{examples}
\end{figure}
%%%%%%%%%%%%%%%%%%%%%%%%%%%%%%%%%%%%%%%%%%%%%%%%%%%%%%%%%%
\section{Structural properties}
In this section we derive several important quantities of the RCG, showing that it is an appropriate model for the small-world complex networks. Thanks to the deterministic feature of $C_q(g)$, we can give exact expressions for the properties of the graph. We will give the explicit result for its order, size, degree distribution, degree correlation, average distance, clustering coefficient, number of spanning tree and Kirchhoff index.

%\subsection{Order and Size}

We denote by $N(g)$ and $M(g)$ respectively the order and the size of $C_q(g)$. Next we show how to derive these quantities. Assume that the number of vertices and the number of edges that get newly generated at step $g$ is denoted as $L_V(g)$ and $L_E(g)$. Then it is obvious that we have $L_V(g)=q\times N(g-1)$ for $g\geqslant 1$, which leads to the result of $N(g)$ along with the initial condition $N(0)=q$. With respect to the size of the network, we have $L_E(g)=q(q-1)/2\cdot N(g-1)+q\cdot N(g-1)$, $g\geqslant 1$, and the initial condition $M(0)=q(q-1)/2$.

\begin{proposition}
The order and size of the graph $C_q(g)=(V(g),E(g))$ are, respectively.
\begin{equation}
N(g)=q(q+1)^g
\end{equation}
and
\begin{equation}
M(g)=\frac{1}{2} q \left((q+1)^{g+1}-2\right)
\end{equation}
\end{proposition}

The average degree is $\bar{\delta}(g)=-2 (q+1)^{-g}+q+1$ which tends to $q+1$ for large $g$. Note that many real-life networks are sparse and their average degree tend to a constant value.

\subsection{Degree distribution}

The degree distribution ${\cal{P}}(\delta)$ for a network is a function indicating the fraction of vertices with degree $\delta$ over all vertices. The degree distribution is a very important characteristic of a graph. It is essential to the analysis of many other structural properties.% Many networks in real life has a power-law distribution while some others follow exponential distribution.

The cumulative degree distribution~\cite{Ne03} is defined as
$$\Pcum(\delta)=\sum_{\delta'=\delta}^{\infty}{\cal{P}}(\delta')\,,$$
which is often used to analyse the degree distribution of a graph. The quantity gives the fraction of vertices whose degree $\delta^{\prime}$ is greater than or equal to $\delta$. In addition, networks whose degree distributions are exponential:
${\cal{P}}(\delta) \sim e^{-\alpha\delta}$, have also an exponential
cumulative distribution with the same exponent:
\begin{equation} \label{cum-deg-dist}
\Pcum(\delta)=\sum_{\delta'=\delta}^{\infty}P(\delta')\approx
\sum_{\delta'=\delta}^{\infty}e^{-\alpha\delta'}=\left(\frac{e^{\alpha}}{e^{\alpha}-1}\right)
e^{-\alpha\delta}\,.
\end{equation}

Next we investigate the degree distribution of $C_q(g)$. We find that at time $g=0$, the network has $q$ vertices of degree $q-1$. Now we study the degree of some vertex $v$ at step $g$. Let the value be $\delta_v(g)$, we look in details about how the quantity evolves. We assume that vertex $v$ is added to the network at step $g_v$, ($g_v\geqslant 0$). For any $g_v>0$, there are $\delta_v(g_v)=q$, where $q-1$ edges link to other vertices in $K_q$ and the other $1$ edge links to $C_q(g-1)$. At every step every existed vertex increase its degree by $q$.

\begin{theorem}
The cumulative degree distribution of the graph $C_q(g)$ follows an exponential distribution: $\Pcum(\delta)\sim (q+1)^{-\frac{\delta}{q}+1}$.
\end{theorem}
\begin{proof}
The degree of vertex $v$ at step $g$, denoted as $\delta_v(g)$, can be written as
\begin{equation}
\delta_v(g+1)=\delta_v(g)+q\,.
\end{equation}
Thus we have
\begin{eqnarray}
\label{degGen}
\delta_v(g)&=&q\cdot(g-g_v+1)\,,(g_v>0)\,,\,and\nonumber\\
\delta_v(g)&=&q\cdot(g+1)-1\,,(g_v=0)\,.
\end{eqnarray}
This means that the numbers of vertices with the degree equal to $q,2q,\cdots,g q,(g+1)q-1$ are, respectively, $q^2 (q+1)^{g-1},q^2 (q+1)^{g-2},\cdots,q^2,q$. %And all these vertices construct the whole network $C_q(g)$.

For a certain value of degree $\delta$, we have $\Pcum(\delta)=\left(N(0)+\sum_{g^\prime=1}^\theta L_V(g^\prime)\right)/N(g)$ where $\theta=\lfloor g-(\delta+q)/q\rfloor$. % Thus $$\Pcum(\delta)=\sum_{g^\prime \leqslant\theta}{\cal{P}}(g_v=g^\prime)\,.$$
Therefore we can find
\begin{eqnarray}
\Pcum(\delta)&=&\frac{1}{q(q+1)^g}\left(q+\sum_{g^\prime=1}^{\theta}{q^2(q+1)^{g^\prime-1}}\right)\nonumber\\
 &=&\frac{1}{(q+1)^{g}}+\frac{q}{(q+1)^g}\sum_{g^\prime=1}^{\theta}(q+1)^{g^\prime-1}\nonumber\\
 &=&(q+1)^{\theta-g}\,.
\end{eqnarray}
For large $g$ we have $$\Pcum(\delta)= (q+1)^{\theta-g}\sim (q+1)^{-\frac{\delta}{q}+1}\,. $$
\end{proof}

\subsection{Degree correlation}
One important parameter for the degree correlation is the average degree of adjacent vertices of all vertices which is referred to as any vertex with degree $\delta$. We denote the parameter by $k_{\rm nn}(\delta)$. If $k_{\rm nn}(\delta)$ increases with $\delta$, this means that the vertices have a tendency to connect to vertices with a similar or larger degree. In this case we claim the graph to be assortative.
The considered value of vertices with degree $\delta$ which is generated at step $g_v$ can be written as
\begin{small}
\begin{equation}
\begin{split}
k_{\rm nn}(\delta)=&\frac{1}{L_V(g_v)\cdot q(g-g_v+1)}\Bigg(q^2(q(g+1)-1)+\sum_{g_p=1}^{g_v-1}L_v(g_p)q\cdot q(g-g_p+1)\\
&+L_v(g_v)(q-1)q(g-g_v+1)+L_v(g_v)\sum_{i=g_v+1}^{g}q\cdot q(g-i+1)\Bigg)\\
=&\frac{1}{2} q (g-g_v+2)+\frac{1+q-2(1+q)^{-g_v+1}}{q(g-g_v+1)}\,.
\end{split}
\end{equation}
\end{small}
According to Eq.(\ref{degGen}), we can express it as
\begin{equation}
k_{\rm nn}(\delta)=\frac{1}{2}(q+\delta)+\frac{q+1-2(1+q)^{-g+\frac{\delta}{q}}}{\delta}\,.\nonumber
\end{equation}
Since $q\leqslant \delta \leqslant q(g+1)-1$, therefore we have $0<(1+q)^{\frac{\delta}{q}-g}<1$, which means
\begin{equation}
k_{\rm nn}(\delta)\approx\frac{1}{2}(q+\delta)+\frac{q}{\delta}\,.
\end{equation}
As for the initial vertices, everyone of them has the same distribution on the degrees of its neighbors. It leads to
\begin{equation}
\begin{split}
k_{\rm nn}(\delta_0)=&\frac{1}{q\cdot(g+1)-1}\left((q-1)(q(g+1)-1)+\sum_{i=1}^g q\cdot q(g-i+1)\right)\\
=&\frac{(g+2)q}{2}-\frac{(g+2) q-2}{2 (g q+q-1)}\,,
\end{split}
\end{equation}
which yields
\begin{equation}
k_{\rm nn}(\delta_0)=\frac{1}{2}\left(\delta_0+q+\frac{1-q}{\delta_0}\right)\,.
\end{equation}
By checking the results we can see that the considered graph is assortative.

\subsection{Average distance}
Given a graph $G=(V,E)$, its {\em average distance} or {\em mean
distance} is defined as: $\mu (G) = \frac{1}{|V(G)|(|V(G)|-1)}
\cdot\sum_{u,v \in V(G)} d(u,v)$ where $d(u,v)$ is the distance between the pair of vertices $u$ and $v$.

\begin{theorem}
The average distance of graph $C_q(g)$ is
\begin{small}
\begin{equation}\label{muDistance}
\mu(C_q(g))=\frac{(q+1)^{-g} \left(2 g q^2 (q+1)^{2 g-1}+(q+1)^g+(q-2)
   \left((q+1)^2\right)^g\right)}{q (q+1)^g-1}\,.
\end{equation}
\end{small}
\end{theorem}

\begin{proof}
To begin with, we assume that the summation of distances between all pairs of vertices in $C_q(g)$ is $D(C_q(g))$. The sum of distances between all pairs $(u,v)$ where $u$ belongs to vertex set $U$ and $v$ belongs to a disjoint vertex set $V$, is denoted as $D(U,V)$.

In order to utilize the recursive construction process to the recursive corona graph we classify the vertex pairs in $C_q(g)$ into 4 different categories $W$, $X$, $Y$ and $Z$. The sum of the distances for the 4 categories is denoted as $S_W$, $S_X$, $S_Y$ and $S_Z$, respectively.

Category $W$ refers to the pairs within the same $K_q$ that we add to the network at step $g$. Category $X$ refers to the pairs $u,v$ where $u$ is selected from one of the $N(g-1)$ $K_q$s added at step $g$ and $v$ selected from any other $K_q$ added to the network at the same step. $Y$ refers to the pairs where $u$ is a new vertex and $v$ is a vertex in $C_q(g-1)$. As for category $Z$, it indicates the pairs where both $u$ and $v$ are from the previous generation of graph $C_q(g-1)$.
Thus we have the following equations:
\begin{eqnarray}
D(C_q(g))&=&S_W(g)+S_X(g)+S_Y(g)+S_Z(g)\,,\\
S_W(g)&=&\frac{q(q-1)}{2}\cdot N(g-1)\,,\\
S_X(g)&=&\sum_{i,j\in V(C_q(g-1)) \atop i<j}(d_{i,j}+2)\cdot q^2\nonumber\\
&=&q^2 S_Z(g)+q^2N(g-1)\cdot(N(g-1)-1)\,,\\
S_Y(g)&=&\sum_{i,j\in V(C_q(g-1))}(d_{i,j}+1)\cdot q\nonumber\\
&=&2 q S_Z(g)+q N(g-1)\cdot(N(g-1)-1)+N(g-1)\cdot q\,,\\
S_Z(g)&=&D(C_q(g-1))\,.
\end{eqnarray}
Combining these recursive expression together we have:
\begin{equation}
   D(C_q(g))=(q+1)^2 D(C_q(g-1))+\frac{1}{2} q^2 \left(2 q (q+1)^{g-1}-1\right)
   (q+1)^g\,,
\end{equation}
with the initial condition $D(C_q(0))=q(q-1)/2$ we get the result
\begin{equation}
D(C_q(g))=\frac{1}{2} q \left(2 g q^2 (q+1)^{2 g-1}+(q+1)^g+(q-2)
   \left((q+1)^2\right)^g\right)\,.
\end{equation}
Dividing $D(C_q(g))$ by $\frac{N(g)(N(g)-1)}{2}$ yields Eq.(\ref{muDistance}). For large $g$, we have $\mu(C_q(g))\sim 2g\sim 2\log_q{N(g)}$, which increases logarithmically with the network order.
%\begin{small}
%\begin{equation}
%\mu(C_q(g))=\frac{(q+1)^{-g-1} \left(\left(q \left(2 g q^2
%   (q+1)^g+q+3\right)-2\right) (q+1)^g+(q ((q-1)
%   q-4)+2) \left((q+1)^2\right)^g\right)}{r \left(q
%   (q+1)^g-1\right)}
%\end{equation}
%\end{small}

\end{proof}

\subsection{Clustering coefficient}
%%need to be override
%The clustering coefficient of a graph is another parameter used to
%characterize small-world networks. The clustering coefficient of a
%vertex was introduced in~\cite{WaSt98} to quantify this concept:
%Given a graph $G=(V,E)$, for each vertex $v\in V(G)$ with degree
%$\delta_v$, its {\em clustering coefficient} $c(v)$ is defined as
%the fraction of the ${\delta_v\choose 2}$ possible edges among the
%neighbors of $v$ that are present in $G$. More precisely, if
%$\epsilon_v$ is the number of edges between the $\delta_v$ vertices
%adjacent to vertex $v$, its {\em clustering coefficient} is
%\begin{equation}
%\label{c(v)}
%c(v)=\frac{2\epsilon_v}{\delta_v(\delta_v-1)},
%\end{equation}
%whereas the {\em network clustering coefficient} of $G$, denoted by $C(G)$,
%is the average of $c(v)$ over all  vertices $v$ of $G$:
%\begin{equation}
%\label{c(G)} c(G)=\frac{1}{|V(G)|}\sum_{v\in V(G)}c(v).
%\end{equation}

Clustering coefficient~\cite{WaSt98} is another crucial quantity used to characterize network structure. Many works about determining clustering coefficient and its related quantities are done on both graph models and graphs in reality~\cite{WaSt98,HoKi02,Ts08,Ts11}.

The clustering coefficient of vertex $v$ is defined as the following quantity
\begin{equation}
\label{c(v)}
c(v)=\frac{2\epsilon_v}{\delta_v(\delta_v-1)}\,,
\end{equation}
where $\epsilon_v$ is the number of edges between the neighbours of vertex $v$. The network clustering coefficient $C(G)$ is defined as the average of $c(v)$ among all vertices. That is,
\begin{equation}
\label{c(G)} c(G)=\frac{1}{|V(G)|}\sum_{v\in V(G)}c(v)\,.
\end{equation}

\begin{theorem}
\label{vClustering}
Let $v(\delta)$  be a vertex in $C_q(g)$ whose degree is $\delta$. Except the initial vertices, its clustering coefficient is
\begin{equation}
c(v(\delta))=\frac{q-1}{\delta-1}\approx\frac{q-1}{\delta}\,.
\end{equation}
\end{theorem}

\begin{proof}
Let us review the intermediate result in calculating the degree distribution that the number of vertices with degree $q,2q,\cdots,g q,(g+1)q-1$ are, respectively, $q^2 (q+1)^{g-1},q^2 (q+1)^{g-2},\cdots,q^2,q$. Except the initial vertices, the clustering coefficient of other vertices follow the same rule, that is, a vertex with degree $k\cdot q$ has $k\cdot q$ neighbours, which are evenly distributed in $k$ clusters. Each cluster forms a complete graph $K_q$. Thus the clustering coefficient of vertex $v$ is derived as
\begin{equation}
c(v)=\frac{k q(q-1)/2}{kq(kq-1)/2}=\frac{q-1}{kq-1},\,(g_v>0)\,,
\end{equation}
and for initial vertices
\begin{equation}
\label{clusterInter}
c(v)=\frac{(q-1)(q-2)+gq(q-1)}{(q(g+1)-1)(q(g+1)-2)},\,(g_v=0)\,.
\end{equation}
Theorem~\ref{vClustering} is naturally gained.
\end{proof}

\begin{theorem}
\label{gClustering}
The clustering coefficient of RCG network $C_q(g)$ converge to
\begin{equation}
c(C_q(g))\sim \frac{q-1}{q+1} \Phi
   \left(\frac{1}{q+1},1,\frac{q-1}{q}\right)\,,
\end{equation}
when the network order is high enough. In the expression, $\Phi$ is the Lerch transcendent function. For large $q$, the clustering coefficient tends to $1$.
\end{theorem}

\begin{proof}
From Eq.(\ref{c(G)}) and (\ref{clusterInter}) we can obtain the clustering coefficient of $C_q(g)$
\begin{small}
\begin{eqnarray}
c(C_q(g))&=&\frac{1}{N(g)}\left[\sum_{k=1}^g {\frac{q-1}{kq-1}\cdot q^2(q-1)^{g-k}}+\frac{(q-1)(gq+q-2)}{(q(g+1)-1)(q(g+1)-2)}\right]\nonumber\\
&=&\frac{q-1}{q+1}
\left(
\Phi\Bigg(\frac{1}{q+1},1,\frac{q-1}{q}\right)\\ \nonumber
& &-\left(\frac{1}{q+1}\right)^g \Phi
   \left(\frac{1}{q+1},1,\frac{qg+q-1}{q}\right)+\frac{(q+1)^{-g+1}
   }{g q+q-1}\Bigg)\,.
\end{eqnarray}
\end{small}
When $g$ is large enough, we obtain
\begin{equation}
c(C_q(g))\sim \frac{q-1}{q+1} \Phi
   \left(\frac{1}{q+1},1,\frac{q-1}{q}\right)\,,
\end{equation}
which is high and tends to $1$ when $q$ is also large. Figure~\ref{clusterCoef} gives the  clustering coefficient of some networks while they grow.
\end{proof}

%%%%%%%%%%%%%%%%%%%%%%%%%%%%%%%%%%%%%%%%%%%%%%%%%%%%%%%%%
% Figure  1
%%%%%%%%%%%%%%%%%%%%%%%%%%%%%%%%%%%%%%%%%%%%%%%%%%%%%%%%%%
\begin{figure}[htbp]
\begin{center}
\includegraphics[width=.85\linewidth]{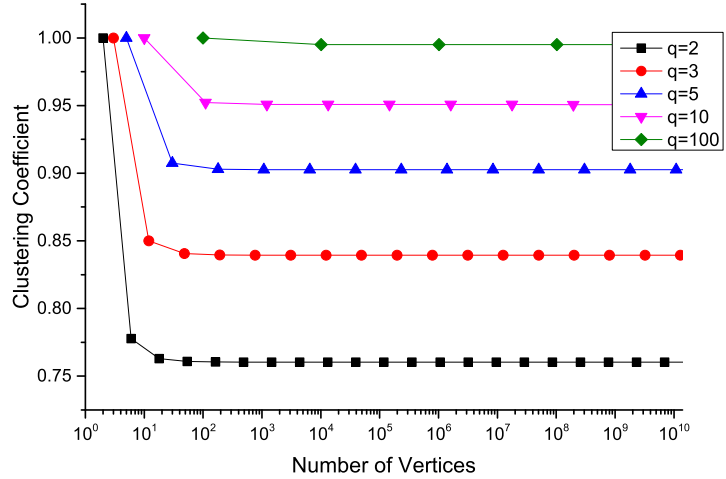}
\end{center}
\caption[kurzform]{Clustering coefficient of some recursive corona graphs.}\label{clusterCoef}
\end{figure}
%%%%%%%%%%%%%%%%%%%%%%%%%%%%%%%%%%%%%%%%%%%%%%%%%%%%%%%%%%

%\subsection{Pearson Correlation}
%The Pearson correlation of a given graph $G(V,E)$ is expressed as
%\begin{equation}
%r=\frac{|E(t)|\sum_i j_i k_i-\left(\sum_i\frac{1}{2}(j_i+k_i)\right)^2}{|E(t)|\sum_i\frac{1}{2} (j_i^2+k_i^2)-\left(\sum_i\frac{1}{2}(j_i+k_i)\right)^2}\,,
%\end{equation}
%where $j_i$ and $k_i$ are the degree of the endvertices of the $i$th edge, and every summation is over all edges in the graph.
%
%Note that we have to calculate $\sum_i j_i k_i$, $\sum_i(j_i+k_i)$ and $\sum_i (j_i^2+k_i^2)$ over all edges.
\subsection{Spanning trees}

Next we derive the number of spanning trees in graph $C_q(g)$.

\begin{theorem}
\label{spanningT}
The number of spanning tree of $C_q(g)$ is
\begin{equation}\label{spanningTE1}
N_\textrm{\rm{tr}}(C_q(g))=q^{q-2} (q+1)^{1-q} \left((q+1)^{q-1}\right)^{(q+1)^g}\,.
\end{equation}
\end{theorem}
\begin{proof}
According to the Cayley's formula~\cite{Ca89}, the number of spanning trees of a complete graph $K_q$ is equal to $q^{q-2}$. Since all vertices of a $K_q$ added to the graph are connected to a vertex in the original graph, these $q+1$ vertices consist of a new complete graph $K_{q+1}$. Therefore we have the following recursive relation of the spanning trees of $C_q(g)$:
\begin{equation}
N_\textrm{\rm{tr}}(C_q(g))=\left((q+1)^{q-1}\right)^{N(g-1)}N_\textrm{tr}(g-1)\,.
\end{equation}
Together with the initial condition $N_{tr}=q^{q-2}$, we can derive the expression for $N_{tr}$;
\begin{equation}
N_\textrm{\rm{tr}}(C_q(g))=q^{q-2} (q+1)^{1-q} \left((q+1)^{q-1}\right)^{(q+1)^g}\,.
\end{equation}
\end{proof}

\subsection{Kirchhoff index}
Resistance distance is an important character of a graph, which can imply many of its dynamic properties. The Kirchhoff index~\cite{KlRa93} of a graph refers to the sum of resistance between all vertex pairs in an associated electrical network obtained from the graph by replacing each edge of the graph by a unit resistance. Denote the effective resistance between vertices $i$ and $j$ as $r(i,j)$ or $r_{i,j}$, then the Kirchhoff index $R_\textrm{Kr}$ of graph $G$ is defined as
\begin{equation}
R_\textrm{Kr}(G)=\sum_{i,j\in V(G) \atop{i<j}} r_{i,j}\,.
\end{equation}
 We denote the Kirchhoff index of graph $C_q(g)$ by $R_\textrm{Kr}(C_q(g))$.

\begin{theorem}
\label{kirchhoff}
The Kirchhoff index of $C_q(g)$ is
\begin{equation}
R_\textrm{\rm{Kr}}(C_q(g))=\left(q^3(2g+1)-2q-1\right)(q+1)^{2g-2}+q(q+1)^{g-1}\,.
\end{equation}
\end{theorem}
\begin{proof}
We denote by $r(U,V)$ as the sum of all effective resistance between pairs $(u,v)$ in which $u$ and $v$ belong to two disjoint vertex set $U$ and $V$ respectively. Similar to the method we used in calculating the average distance, we classify these pairs into $4$ categories $W$, $X$, $Y$ and $Z$, where the definition is exactly the same as used in calculating the average distance. Then the sum of the distances for the four categories is denoted as $R_W$, $R_X$, $R_Y$ and $R_Z$. We have the following equations:
\begin{small}
\begin{eqnarray}
R_\textrm{Kr}(C_q(g))&=&R_W(g)+R_X(g)+R_Y(g)+R_Z(g)\,,\\
R_W(g)&=&\frac{q(q-1)}{2}\cdot\frac{2}{q+1}\cdot N(g-1)\,,\\
R_X(g)&=&\sum_{i,j\in V(C_q(g-1)) \atop i<j}(r_{i,j}+2\cdot\frac{2}{q+1})\cdot q^2\nonumber\\
&=&q^2 R_Z(g)+\frac{2\cdot q^2N(g-1)\cdot(N(g-1)-1)}{q+1}\,,\\
R_Y(g)&=&\sum_{i,j\in V(C_q(g-1))}(r_{i,j}+\frac{2}{q+1})\cdot q\nonumber\\
&=&2 q R_Z(g)+q\cdot\frac{2}{q+1} N(g-1)\cdot(N(g-1)-1)\nonumber\\
& &+q\cdot\frac{2}{q+1}\cdot N(g-1)\,,\\
R_Z(g)&=&R_\textrm{Kr}(C_q(g-1))\,,
\end{eqnarray}
\end{small}
which yields
\begin{equation}
R_\textrm{Kr}(C_q(g+1))=q^2 \left(2 q (q+1)^g-1\right) (q+1)^g+(q+1)^2 R_\textrm{Kr}(C_q(g))\,.
\end{equation}
Notice that the effective resistance between vertices $v$ and $u$ in a complete graph $K_q$ is $2/q$ since the potential between any other vertices are identical, if we impose potential difference between $u$ and $v$.

Along with the initial condition $R_k(C_q(0))=q-1$ we can deduce
\begin{eqnarray}
R_\textrm{Kr}(C_q(g))&=&\left(q^3-2 q-1\right) (q+1)^{2g-2}+q \left(2 q^2 g (q+1)^g+q+1\right)
   (q+1)^{g-2}\nonumber\\
   &=&(q^3-2q-1+2gq^3)(q+1)^{2g-2}+(q+1)^{g-1}q\nonumber\\
   &=&\left(q^3(2g+1)-2q-1\right)(q+1)^{2g-2}+q(q+1)^{g-1}\,.
\end{eqnarray}
This completes the proof.
\end{proof}
\section{Spectral analysis}
By convention the (unweighted) adjacency matrix $A(G)$ of a graph $G$ is defined as a $N\times N$ matrix with the entry $a_{i,j}$ representing the number of edges incident with endpoints ${i,j}$. The degree matrix  of $G$, $D(G)$ is defined as a diagonal matrix with its $i$th entry on the main diagonal equal to the degree of vertex $i$. We call $L(G)=D(G)-A(G)$ the Laplacian matrix of graph $G$. These matrices determine the structure of the graph, and the eigenvalues of $A(G)$ and $L(G)$ are sensitive to many of the structural properties, which have remarkable impact on the dynamic processes superimposed upon the network.

\begin{definition}
Given $A(C_q(g))$ the adjacency matrix of $C_q(g)$, we define the \emph{spectra} of $C_q(g)$ as
\begin{equation}
\sigma(C_q(g)):=\sigma(g)=(\lambda_1(g),\lambda_2(g)\cdots\lambda_n(g))\,.
\end{equation}
\end{definition}
Similarly, we have:
\begin{definition}
Given $L(C_q(g))$ as the Laplacian matrix of $C_q(g)$, we define its \emph{Laplacian spectra}
\begin{equation}
S(C_q(g)):=S(g)=(\gamma_1(g),\gamma_2(g),\cdots,\gamma_n(g))\,.
\end{equation}
\end{definition}

\subsection{Spectra of Adjacency Matrix}
\begin{theorem}
\label{eigAdj}
The relation between $\sigma(g)$ and $\sigma(g-1)$ is
\begin{enumerate}
\item $\frac{\lambda_i(g-1)+q-1\pm \sqrt{(q-1-\lambda_i(g-1))^2+4q}}{2}\in \sigma(g)$ with multiplicity $1$ for $i=1,\cdots,q(q+1)^{g-1}$ and
\item $-1 \in \sigma(g)$ with multiplicity $(q-1)q(q+1)^{g-1}$.
\end{enumerate}
\end{theorem}
The result is a corollary of results in~\cite{BaPaSa07,ShAdMi15}.

\subsection{Spectra of Laplacian Matrix}
\begin{theorem}\label{theoremLap}
The relation between $S(g)$ and $S(g-1)$ is
\begin{enumerate}
\item $\frac{\gamma_i(g-1)+q+1\pm \sqrt{(\gamma_i(g-1)+q+1)^2-4\gamma_i(g-1)}}{2}\in S(g)$ with multiplicity $1$ for $i=1,\cdots,q(q+1)^{g-1}$ and
\item $q+1\in S(g)$ with extra multiplicity $(q-1)q(q+1)^{g-1}$.
\end{enumerate}
\end{theorem}

Note that in the first part $\gamma_i(g-1)=0$ will generate an eigenvalue equal to $q+1$ with multiplicity $1$ in iteration $g$. So the actual multiplicity of $q+1$ is $(q-1)q(q+1)^{g-1}+1$ for any $g\geqslant 1$. The proof of Theorem~\ref{theoremLap} is evident using methods in~\cite{BaPaSa07,Li14,ShAdMi15}.  For convenience of the following discussion we give a similar proof here:

\begin{proof}
The Laplacian matrix of $C_q(g)$ is
\begin{equation}
L(C_q(g))=\left(
\begin{tabular}{c|ccc}
   $L(C_q(g-1))+q I_n$ & $-I_n$ & $\cdots$ & $-I_n$ \\
   \hline
   $-I_n$ &  &  &  \\
   $\vdots$ &  & $\left(L(K_q)+I_q\right)\otimes I_n$ &  \\
   $-I_n$ &  &  &  \\
\end{tabular}\right)
\end{equation}

Let $Y_1,\cdots,Y_{N(g-1)}$ be the Laplacian eigenvectors of $C_q(g-1)$ corresponding to the eigenvalues $\gamma_1(g-1),\gamma_2(g-1),\cdots,\gamma_{N(g-1)}(g-1)$, respectively. For $i=1,\cdots, N(g-1)$, let
\begin{small}
\begin{equation}
\begin{split}
\phi_i=&\frac{\gamma_i(g-1)+q+1+ \sqrt{(\gamma_i(g-1)+q+1)^2-4\gamma_i(g-1)}}{2}\,,\\
\hat{\phi}_i=&\frac{\gamma_i(g-1)+q+1- \sqrt{(\gamma_i(g-1)+q+1)^2-4\gamma_i(g-1)}}{2}\,.\nonumber
\end{split}
\end{equation}
\end{small}
Note that $\phi_i$, $\hat{\phi}_i$ are Laplacian eigenvalues of $C_q(g)$ corresponding to the eigenvectors
\begin{equation}
\left(
  \begin{array}{c}
     Y_i \\
     g(\phi_i)Y_i \\
     \vdots \\
     g(\phi_i)Y_i \\
  \end{array}
\right)\,,\left(
  \begin{array}{c}
     Y_i \\
     g(\hat{\phi}_i)Y_i \\
     \vdots \\
     g(\hat{\phi}_i)Y_i \\
  \end{array}
\right)\,,\nonumber
\end{equation}
respectively.
 In fact $\phi_i$ is obtained by solving
\begin{small}
\begin{equation}
\begin{split}
&\phi_i\left(
  \begin{array}{c}
     Y_i \\
     g(\phi_i)Y_i \\
     \vdots \\
     g(\phi_i)Y_i \\
  \end{array}
\right)=L(C_q(g))\left(
  \begin{array}{c}
     Y_i \\
     g(\phi_i)Y_i \\
     \vdots \\
     g(\phi_i)Y_i \\
  \end{array}
\right)\\
&=\left(
  \begin{array}{c}
     L(C_q(g-1))+qI_n \\
     -I_n \\
     \vdots \\
     -I_n \\
  \end{array}
\right)Y_i+\sum_{j=1}^{q}\left(
  \begin{array}{c}
     -I_n \\
     (L(K_q)+I_n)_{1j} I_n \\
     \vdots \\
     (L(K_q)+I_n)_{qj}I_n \\
  \end{array}
\right)g(\phi_i)Y_i\\
&=\left(
  \begin{array}{c}
     \gamma_i(g-1)+q-q\cdot g(\phi_i) \\
     -1+g(\phi_i) \\
     \vdots \\
     -1+g(\phi_i)\\
  \end{array}
\right)Y_i\,.
\end{split}
\end{equation}
\end{small}
Thus we can derive the following equations:
\begin{eqnarray}
\gamma_i(g-1)+q-q\cdot g(\phi_i)&=&\phi_i \,,\label{laIt1}\\
-1+g(\phi_i)&=&\phi_i\cdot g(\phi_i)\,.\label{laIt2}
\end{eqnarray}

From Eq.(\ref{laIt2}) we can obtain that $g(\phi_i)\neq 0$. Therefore we can substitute Eq.(\ref{laIt2}) into Eq.(\ref{laIt1}), we have:
\begin{equation}\label{lapEig}
\left(\gamma_i(g-1)+q-\phi_i\right)(\phi_i-1)=-q\,,
\end{equation}
which leads to the result of the first part of the theorem.

If the Laplacian eigenvalues $\nu_1=0, \nu_2=\nu_3=\cdots=\nu_q=q$ of $L(K_q)$ are correlated with the eigenvectors $Z_1,Z_2,\cdots,Z_q$, respectively, then for $j=2,\cdots,q$, we have
\begin{equation}
L(C_q(g))\left(
  \begin{array}{c}
     \mathbf{0} \\
     Z_j\otimes e_i \\
  \end{array}
\right)=(q+1)\left(
  \begin{array}{c}
     \mathbf{0} \\
     Z_j\otimes e_i \\
  \end{array}
\right)\,.
\end{equation}
This completes the proof.
\end{proof}

%\begin{theorem}
%The number of spanning tree of $C_q(g)$ is
%\begin{equation}\label{spanning}
%N_{\textrm{tr}}=q^{q-2} (q+1)^{(q-1) \left((q+1)^g-1\right)}\,.
%\end{equation}
%\end{theorem}
Next we use the results of the Laplacian spectra to prove Theorem~\ref{spanningT} and \ref{kirchhoff}. First we give an alternative proof of Theorem~\ref{spanningT}.

\begin{proof}
It is known that the number of spanning tree of a graph $G$ has the following form~\cite{Bi93,TzWu00}
\begin{equation}
N_{\textrm{tr}}(G)=\frac{\prod_{i=2}^N \tau_i}{N}\,,
\end{equation}
where $N$ is the number of vertices and $\tau_i$ refers to $N$ eigenvalues of the graph $G$. Given the graph is connected, let $\tau_1$ be the unique zero eigenvalue, then $\tau_i$, $i=2,\cdots,N$ are $N-1$ non-zero eigenvalues of the graph Laplacian.

Theorem \ref{theoremLap} tells that the Laplacian spectrum of $C_q(g)$ consists of two parts. For the first part, we can derive from Eq.(\ref{lapEig}) that, in iteration $g$, $\gamma_i$ in $C_q(g-1)$ generates two eigenvalues $\phi_i$ and $\hat\phi_i$ which are subject to the relations $\phi_i\hat\phi_i=\gamma_i$ and $\phi_i+\hat\phi_i=\gamma_i+q+1$. In particular the trivial eigenvalue $\gamma_1=0$ generates $\phi_i=q+1$ and $\hat\phi_i=0$. As for the second part, there is an eigenvalue $\gamma=q+1$ with multiplicity $(q-1)q(q+1)^{g-1}$.% Therefore the number of spanning trees of the considered graph is obtained according to Eq.(\ref{spanning}).
 We denote by $S(g)$ the sum of all non-zero eigenvalues of $L(C_q(g))$ and by $\Upsilon(g)$ the product of all non-zero eigenvalues of $L(C_q(g))$. Then we can obtain
\begin{small}
\begin{equation}\label{prodRe}
\begin{split}
\Upsilon(g)&=\prod_{i=2}^{N(g)}\gamma_i(g)=(q+1)^{(q-1)q(q+1)^{g-1}+1}\prod_{i=2}^{N(g-1)}\phi_i\hat\phi_i\\
&=(q+1)^{(q-1)q(q+1)^{g-1}+1}\prod_{i=2}^{N(g-1)}\gamma_i(g-1)\\
&=\Upsilon(g-1)\cdot q(q+1)^{(q-1)q(q+1)^{g-1}+1}\,.
\end{split}
\end{equation}
\end{small}
Eq.(\ref{prodRe}) and the initial condition $\Upsilon(0)=q^{q-1}$ yield
\begin{equation}
\Upsilon(g)=q^{q-1} (q+1)^{(q-1) \left((q+1)^g-1\right)+g}\,.\label{prodRS}
\end{equation}
Therefore
\begin{equation}
N_{\textrm{tr}}(C_q(g))=q^{q-2} (q+1)^{(q-1) \left((q+1)^g-1\right)}\,.\nonumber
\end{equation}
\end{proof}
The result is equivalent to what we derived using combinatorial method.

%\begin{theorem}
%The Kirchhoff index of the graph is
%\begin{equation}
%R_{\textrm{Kr}}(g)=(q+1)^{g-2} \left(\left((2 g+1) q^3-2 q-1\right) (q+1)^g+q (q+1)\right)\,.
%\end{equation}
%\end{theorem}
In the following we give an alternative proof of Theorem \ref{kirchhoff} using the spectra information.
\begin{proof}
The Kirchhoff index of a graph $G$ can be expressed as~\cite{AlFi02,GhBoSa08}:
\begin{equation}
R_\textrm{Kr}(G)=N\cdot\sum_{i=2}^N\frac{1}{\tau_i}\,,
\end{equation}
where $N$ and $\tau_i$ are the same as the previous definition.
Let
\begin{equation}
S=\sum_{j=2}^{N} \prod_{{i=2}\atop{i\neq j}}^{N}\gamma_i=\prod_{i=2}^{N}\gamma_i\cdot\sum_{i=2}^{N}\frac{1}{\gamma_i}\,.
\end{equation}
 We can follow the clue of the previous analysis by separating the eigenvalues of its Laplacian matrix into two parts. Recall that the eigenvalues of the Laplacian consist of two parts $\Gamma^{(1)}$ and $\Gamma^{(2)}$ as defined by Theorem \ref{theoremLap}. Assume that $\Gamma^{(1)\prime}=\Gamma^{(1)}\setminus\{0\}$. For the first part of the eigenvalues of $L(C_q(g))$, we denote them as $\phi_i$ and $\hat\phi_i$, $i=1,2,\cdots,N(g-1)$. Suppose the original eigenvalue in $L(C_q(g-1))$, which is correlated with $\phi_i$ and $\hat\phi_i$, is $\gamma_i$. Then
\begin{equation}
S(g)=S^{(1)}\Upsilon^{(2)}+S^{(2)}\Upsilon^{(1)}\,,
\end{equation}
where
\begin{equation}
S^{(1)}=\sum_{\gamma_j\in\Gamma^{(1)\prime}} \prod_{\gamma_i\in\Gamma^{(1)\prime}\atop{i\neq j}}\gamma_i\,,
\end{equation}
\begin{equation}
\gamma^{(2)}=\prod_{\gamma_i\in\Gamma^{(2)}}\gamma_i\,,
\end{equation}
\begin{equation}
S^{(2)}=\sum_{\gamma_j\in\Gamma^{(2)}} \prod_{\gamma_i\in\Gamma^{(2)}\atop{i\neq j}}\gamma_i\,,
\end{equation}
and
\begin{equation}
\gamma^{(1)}=\prod_{\gamma_i\in\Gamma^{(1)}}\gamma_i\,.
\end{equation}
 Accordingly we can obtain
\begin{small}
\begin{equation}
\begin{split}
S^{(1)}=&(q+1)\prod_{i=2}^{N(g-1)}\phi_i\hat\phi_i\cdot\left(\sum_{j=2}^{N(g-1)}\left(\frac{1}{\phi_j}+\frac{1}{\hat\phi_j}\right)\right)+\prod_{i=2}^{N(g-1)}\phi_i\hat\phi_i\\
=&(q+1)\sum_{j=2}^{N(g-1)}\left(\prod_{{i=2}\atop{i\neq j}}^{N(g-1)}\phi_i\hat\phi_i\cdot\left(\phi_j+\hat\phi_j\right)\right)+\prod_{i=2}^{N(g-1)}\phi_i\hat\phi_i\\
=&\Upsilon(g-1)+(q+1)\sum_{j=2}^{N(g-1)} \prod_{{i=2}\atop{i\neq j}}^{N(g-1)}\gamma_i(g-1)\cdot\left(\gamma_j(g-1)+q+1\right)\\
=&\Upsilon(g-1)+(q+1)\sum_{j=2}^{N(g-1)} \prod_{{i=2}\atop{i\neq j}}^{N(g-1)}\gamma_i(g-1)\cdot\gamma_j(g-1)\\
&+(q+1)\sum_{j=2}^{N(g-1)} \prod_{{i=2}\atop{i\neq j}}^{N(g-1)}\gamma_i(g-1)\cdot\left(q+1\right)\\
=&(1+(q+1)\cdot (q(q+1)^{g-1}-1))\cdot\Upsilon(g-1)+(q+1)^2\cdot S(g-1)\,.
\end{split}
\end{equation}

\begin{equation}
\Upsilon^{(2)}=(q+1)^{(q-1)q(q+1)^{g-1}}\,,
\end{equation}
\begin{equation}
S^{(2)}=(q-1)q(q+1)^{(q-1)q(q+1)^{g-1}+g-2}\,,
\end{equation}
\begin{equation}
\Upsilon^{(1)}=(q+1)\cdot \Upsilon(g-1)\,.
\end{equation}
\end{small}
Therefore we obtain a recursive relation of $S(g)$. Considering the initial condition $S(0)=(q-1)q^{q-2}$, we derive
\begin{small}
\begin{equation}
S(g)=q^{q-2} (q+1)^{(q-1) (q+1)^g+g-q-1} \left(\left((2 g+1) q^3-2 q-1\right) (q+1)^g+q
   (q+1)\right)\,,
\end{equation}
\end{small}
%Thus we need to find the sum of all eigenvalues of $L(C_q(g))$. We denote the quantity as $\Omega(g)$, we have:
%\begin{equation}
%\begin{split}
%\Omega(g)&=\sum_{i=1}^{N(g)} \gamma_i(g)=(q-1)q(q+1)^g+\sum_{i=1}^{N(g-1)}(\phi_i+\hat\phi_i)\\
%&=(q-1)q(q+1)^g+\sum_{i=1}^{N(g-1)}(\gamma_i(g-1)+q+1)\\
%&=(q-1)q(q+1)^g+q(q+1)^{g}+\sum_{i=1}^{N(g-1)}\gamma_i(g-1)\\
%&=q^2(q+1)^g+\Omega(g-1)\,.
%\end{split}
%\end{equation}
thus
\begin{small}
\begin{eqnarray}
\label{spanningTE2}
R_{\textrm{Kr}}(C_q(g))&=&(q+1)^{g-2} \left(\left((2 g+1) q^3-2 q-1\right) (q+1)^g+q (q+1)\right)\nonumber\\
&=&(q^3(2g+1)-2q-1)(q+1)^{2g-2}+q(q+1)^{g-1}\,,
\end{eqnarray}
\end{small}
\end{proof}
Note that Eq.~(\ref{spanningTE2}) is consistent with Eq.~(\ref{spanningTE1}). 
For large $g$, the Kirchhoff index displays the following leading behavior: 
\begin{equation}
R_{\textrm{Kr}}(C_q(g))\sim g\cdot N(g)^2 \sim N(g)^2\log_{q+1}{N(g)}\,,
\end{equation}
%The result is identical to what we derived using the combinatorial method.

\section{Conclusion}
In this paper, we have introduced a deterministically growing  model to generate small-world graph, by using the corona product. The advantage of such a model is that many of its properties can be solved exactly. We have derived explicitly many structural quantities of the small-world model. We have also found the eigenvalues for the adjacency matrix and the Laplacian matrix of the model. In future, the properties of various dynamical processes taking place on the small-world model deserve study.   %with all results given closed form. Compared with the growth of the vertex number, the introduced model shows logarithmic growing average distance, large clustering coefficient, exponentially growing number of spanning trees, and slow growing average effective resistance.
% Such properties make us believe such methods can be used as a network generator operation. %The graph is also expected to fit in some real-world network if we further introduce randomness or sampling subroutine into the scheme.

\section*{Acknowledgement}
This work was supported by the National Natural Science Foundation of China under Grant No.~11275049.

%% The Appendices part is started with the command \appendix;
%% appendix sections are then done as normal sections
%% \appendix

%% \section{}
%% \label{}

%% If you have bibdatabase file and want bibtex to generate the
%% bibitems, please use
%%
%%  \bibliographystyle{elsarticle-num}
%%  \bibliography{<your bibdatabase>}

%% else use the following coding to input the bibitems directly in the
%% TeX file.

%\bibliographystyle{iopart-num}
%\bibliography{network}

%\providecommand{\newblock}{}

%\begin{thebibliography}{00}

%% \bibitem{label}
%% Text of bibliographic item

%\bibitem{}

%\end{thebibliography}
\end{document}